\documentclass[11pt]{article}

\usepackage[a4paper,margin=1in]{geometry}
\usepackage{authblk}
\usepackage{amsmath,amssymb,amsthm,mathtools}
\usepackage{graphicx}
\usepackage{hyperref}
\hypersetup{hidelinks}
\usepackage{enumitem}
\usepackage{xcolor}
\numberwithin{equation}{section}

\newtheorem{theorem}{Theorem}[section]
\newtheorem{proposition}[theorem]{Proposition}
\newtheorem{lemma}[theorem]{Lemma}

\newtheorem{definition}[theorem]{Definition}
\newtheorem{remark}[theorem]{Remark}

\newcommand{\Z}{\mathbb Z}

\title{Exact Results for the Symmetric Dyson Exclusion Process} \author[1]{A. Zahra} \author[2]{J. Dubail} \author[3]{G. M. Sch\"utz} \affil[1]{Laboratoire de Physique et Chimie Th\'eoriques, Universit\'e de Lorraine, Nancy, France } \affil[2]{Centre Europ\'een de Sciences Quantiques and ISIS (UMR 7006), Universit\'e de Strasbourg and CNRS, Strasbourg, France } \affil[3]{Centro de An\'alise Matem\'atica, Geometria e Sistemas Din\^amicos, Departamento de Matem\'atica, Instituto Superior T\'ecnico, Universidade de Lisboa, Lisbon, Portugal } \date{\today}

\begin{document}
\maketitle
 
\begin{abstract}
The symmetric Dyson exclusion process (SDEP) is an exclusion process on the
lattice with a long-range logarithmic Coulomb-type interaction.  It appears in
several equivalent forms: as symmetric random walkers conditioned, in the
Doob-transform sense, not to collide; as the maximal-activity limit of a
conditioned SSEP; and as a ground-state transform of the spin- 1/2 XX
chain.  In this work, we
exploit this latter representation to obtain exact evolution formulas from
deterministic initial configurations. The resulting determinantal kernel
is expressed through a finite interpolation
expression in terms of Lagrange polynomials, leading to explicit density
evolution in finite and infinite lattices. For the melting of a densely packed
block, we show that the full hierarchy of density moments is governed by a
finite-dimensional polynomial algebra, and we identify Catalan numbers in the
leading time coefficients. We also derive the Euler-scale density profile and
the arctic curve separating frozen and liquid regions, and show that they coincide with conjectured hydrodynamic results obtained in a previous work.
\end{abstract}

\tableofcontents

\section{Introduction}
\label{sec:introduction}

Interacting particle systems provide some of the most successful microscopic
models of nonequilibrium statistical mechanics, and the exclusion process is
their paradigm~\cite{Spohn1991,KipnisLandim1999,Liggett2010}.  For systems with
short-range interactions, the passage from the microscopic stochastic dynamics
to a deterministic macroscopic description is by now well understood: under
diffusive scaling the coarse-grained density obeys a closed hydrodynamic
equation---the heat equation for the symmetric simple exclusion process
(SSEP)---as a consequence of the law of large numbers together with local
stationarity~\cite{Spohn1991,KipnisLandim1999}.  In one dimension, integrability
sharpens this picture considerably: for the asymmetric exclusion process the
master equation can be solved exactly~\cite{Schutz1997,Schutz2001}, and the
resulting determinantal, free-fermion structure underlies precise results on
current statistics, Kardar--Parisi--Zhang fluctuations, and limit
shapes~\cite{Johansson2000,BorodinFerrariPrahoferSasamoto2007,ImamuraMucciconiSasamoto2023}.

This well-developed picture relies crucially on the interactions being
short-ranged.  For particle systems with long-range interactions, exact results
are scarce, and even the notion of a local current as a function of the local
density becomes
questionable~\cite{Garrido1990,Liggett1980,AndjelGuiol2005,GoncalvesJara2018,BelitskyNgocSchutz2024}.
A tractable and paradigmatic member of this class is the symmetric Dyson
exclusion process (SDEP): a gas of particles hopping on a ring under exclusion,
with nearest-neighbour rates carrying a long-range, logarithmic (Coulomb-type)
interaction.  
The SDEP arises in several seemingly unrelated contexts.  It was introduced as
a model of vicinal surfaces~\cite{Spohn1999vicinal}, and it also admits an
interpretation as free symmetric walkers conditioned, in the Doob-transform
sense, never to collide.  It appears when the SSEP is conditioned on an
atypically large activity~\cite{PopkovSimonSchutz2010,Schutz2015}; in the limit
of maximal activity, this conditioning generates precisely the logarithmic
interaction.  Its invariant measure is the discrete circular Dyson log-gas,
namely the lattice analogue of the CUE eigenangle distribution
\cite{Dyson1962}.  The model is also closely tied to the Calogero--Sutherland
system~\cite{Calogero1969,Sutherland1972,AbanovBettelheimWiegmann2009}.

The structural fact that makes the SDEP solvable is that its generator is the
ground-state (Doob) transform of the spin-$\tfrac12$ XX quantum chain, a
free-fermion model~\cite{LiebSchultzMattis1961,Niemeijer1967}.  In a companion
paper~\cite{ZahraDubailSchutz2025} we exploited this mapping to argue that, under
ballistic (Eulerian) scaling, the SDEP is governed by a closed but
\emph{non-local} hydrodynamic equation whose current depends on the entire
density profile through a Hilbert transform.
This non-local equation is equivalent to a local two-field complex Hopf
(inviscid Burgers) system of the type arising in free-fermion imaginary-time
hydrodynamics and limit-shape problems
\cite{Abanov2006,KenyonOkounkov2007,PallisterGangardtAbanov2022}, and it predicts limit shapes
and arctic curves that were confirmed by Monte-Carlo simulation.
The real-time dynamics of lattice free fermions launched from zero-entropy
states, including compact-block or double-domain-wall initial conditions, has
itself been studied extensively
\cite{AntalRaczRakosSchutz1999,
BettelheimAbanovWiegmann2006,RuggieroBrunDubail2019,
ScopaEtAl2021,ScopaCalabreseDubail2022}.
In the present work we use the free-fermion mapping not as a route to a
macroscopic conjecture, but as a tool for deriving \emph{exact} results, both at
finite size and in scaling limits.  Our starting point is a dictionary that
expresses any equal-time SDEP observable, evolved from a deterministic initial
configuration, as a free-fermion matrix element taken between the XX ground state
and that configuration.  Since both states are Slater determinants, Wick's
theorem applies and all correlations are determinantal; the relevant kernel is,
however, generically non-Hermitian, because it arises from a mixed matrix element
rather than from an ordinary quantum expectation.  Evaluating this kernel for
deterministic initial data reduces to a trigonometric Lagrange interpolation
problem, which we solve in closed form.  From this single input we obtain exact
finite-volume and infinite-volume formulas for the density profile at all times;
a universal finite-dimensional algebra for the density moments of arbitrary
deterministic finite configurations, whose highest-time coefficients have
Catalan asymptotics; an exact finite representation of the melting of densely
packed blocks and their enhanced one-body spreading; and a rigorous
steepest-descent analysis of the Euler-scaling limit that
yields the emergent limit shape and the boundary---the arctic curve---separating
the frozen regions, where the density is $0$ or $1$, from the liquid region.

These results connect the SDEP to several active lines of research.  The
determinantal correlations place it within the theory of determinantal point
processes that has proved so powerful for exclusion processes and random
matrices. The emergent limit shapes and arctic curves are of the kind familiar
from dimer coverings, the six-vertex model, and random tilings, where they are
analysed through complex Burgers equations and inhomogeneous free-field
techniques~\cite{KenyonOkounkov2007,AllegraDubailStephanViti2016,Gorin2021,ColomoPronko2010,StephanLectureNotes2021,DiFrancescoGuitter2019,PallisterGangardtAbanov2022}.
The block profiles we compute are close relatives of the emptiness formation
probability and the full counting statistics of the XX
chain~\cite{Abanov2006,AbanovFranchini2003,StephanEFP2014,PallisterEtAl2025}.
Finally, in the low-density limit our lattice results reduce to those of the
continuous Dyson gas of Brownian particles with logarithmic
repulsion~\cite{AndrausKatori2016,DandekarKrapivskyMallick2023,DandekarKrapivskyMallick2024,KrapivskyMallick2025},
the hard-core constraint $\rho\le1$ being the essential lattice feature absent in
the continuum.

The paper is organised as follows.
Section~\ref{sec:sdep_xx_transform} recalls the definition of the SDEP---its
hopping rates and generator, the mapping to the XX chain by the ground-state
transform, and the reversible Dyson measure---and fixes the notation used
throughout.  Section~\ref{sec:observable_dictionary} establishes the dictionary
between SDEP observables and free-fermion matrix elements and the resulting
determinantal structure.  Section~\ref{sec:mixed_correlations} evaluates the
mixed correlation kernel for deterministic initial configurations in closed form.
Section~\ref{sec:exact_density_evolution} derives the exact density-evolution
formulas in finite and infinite volume.  Section~\ref{sec:block_melting}
develops the universal polynomial moment algebra for deterministic initial data,
derives the first two moment identities, and specializes the general density
formula to obtain a finite representation and explicit centered moments for a
compact block.
Section~\ref{sec:block_euler_asymptotics} derives the Euler-scale limit shape
and the associated arctic curve.

\section{The SDEP and the XX ground-state transform}
\label{sec:sdep_xx_transform}

We recall the definition of the symmetric Dyson exclusion process (SDEP), following
\cite{ZahraDubailSchutz2025}.  The process is defined on the discrete torus
\(\mathbb T_L=\mathbb Z/L\mathbb Z\), with \(L\geq2\), in a sector with a fixed
number \(N\) of particles, \(1\leq N\leq L\).  Configurations are exclusion
configurations: each site contains at most one particle.  We write a configuration
as the ordered set of occupied sites
\[
    \eta=\{x_1<\cdots<x_N\}\subset\{1,\ldots,L\},
\]
with the usual periodic identification.

The dynamics consists of nearest-neighbour jumps.  A particle at \(x_i\) can move
to \(x_i\pm1\) if the target site is empty; otherwise the move is forbidden.  The
rate of an allowed jump is
\begin{equation}
    c_i^{\pm}(\eta)
    =
    w
    \prod_{j\neq i}
    \frac{
    \sin\left(\frac{\pi(x_i\pm 1-x_j)}{L}\right)
    }{
    \sin\left(\frac{\pi(x_i-x_j)}{L}\right)
    },
    \label{eq:sdep_rate_right}
\end{equation}
where \(w>0\) fixes the microscopic time scale.  Thus the motion is local in
space, since particles only jump to neighbouring sites, but the jump rates are
non-local: the rate of a given particle depends on the positions of all the
others.

A useful way to write the rates is in terms of the positive function
\begin{equation}
    \Phi_N(x_1,\ldots,x_N)
    =
    \prod_{1\leq i<j\leq N}
    \sin\left(\frac{\pi(x_j-x_i)}{L}\right),
    \qquad
    1\leq x_1<\cdots<x_N\leq L .
    \label{eq:Phi_def}
\end{equation}
Indeed, if \(\eta^{x,y}\) denotes the configuration obtained from \(\eta\) by
moving a particle from \(x\) to an empty neighbouring site \(y=x\pm1\), then
\eqref{eq:sdep_rate_right} can be written equivalently as
\begin{equation}
    c(\eta,\eta^{x,y})
    =
    w\,
    \frac{\Phi_N(\eta^{x,y})}{\Phi_N(\eta)},
    \qquad y=x\pm1.
    \label{eq:sdep_rates_abstract}
\end{equation}
This ratio form is independent of the chosen ordering convention and makes the
positivity of the allowed jump rates manifest.
\begin{remark}
\label{rem:noncolliding_walkers}
The ratio form \eqref{eq:sdep_rates_abstract} shows that the SDEP can be viewed
as a Doob transform of free walkers killed at collisions.  Indeed, start from
$N$ independent continuous-time symmetric random walkers on $\mathbb T_L$,
each jumping to each nearest neighbour with rate $w$, and kill the process
when two particles occupy the same site.  The function $\Phi_N$ vanishes on
the collision boundary and is the positive ground state of the killed process.
The Doob transform by $\Phi_N$ changes the free jump rate (w) into

$$w\frac{\Phi_N(\eta^{x,y})}{\Phi_N(\eta)},$$
which is exactly the SDEP rate. Thus the SDEP may be interpreted as the infinite-horizon conditioning of free
walkers to avoid collisions, or equivalently as the associated (Q)-process.
This is the lattice circular analogue of the classical construction of
non-colliding Brownian motions via Doob transforms and
Karlin--McGregor determinants; see, for example,
\cite{Doob1957,KarlinMcGregor1959,KatoriTanemura2007}.
\end{remark}

\subsection{The Markov generator}
\label{subsec:generator}

Formally, the SDEP is the continuous-time Markov process on the configuration
space
\[
    \Omega_{L,N}
    =
    \left\{
    \eta\in\{0,1\}^{\mathbb T_L}:
    \sum_{x\in\mathbb T_L}\eta_x=N
    \right\}
\]
whose generator acts on functions \(f:\Omega_{L,N}\to\mathbb R\) by
\begin{equation}
    (\mathcal L_{L,N}^{\mathrm{SDEP}}f)(\eta)
    =
    \sum_{\substack{x\in\mathbb T_L\\ y=x\pm1}}
    c(\eta,\eta^{x,y})
    \left[
    f(\eta^{x,y})-f(\eta)
    \right],
    \label{eq:sdep_generator}
\end{equation}
with the rates \eqref{eq:sdep_rates_abstract}.  Every result below is obtained by
analysing this generator, and the essential tool is its exact mapping to free
fermions, which we describe next.

\subsection{Free-fermion representation}
\label{subsec:freefermions}

Let \(c_x,c_x^\dagger\) be canonical fermionic annihilation and creation
operators on \(\mathbb T_L\), and let
\begin{equation}
    H_{L}^{XX}
    =
    -w
    \sum_{x=1}^{L}
    \left(
    c_{x+1}^\dagger c_x
    +
    c_x^\dagger c_{x+1}
    \right)
    \label{eq:xx_hamiltonian}
\end{equation}
be the free-fermion XX Hamiltonian, obtained from the spin-\(\tfrac12\) XX chain
by a Jordan--Wigner transformation.  In the \(N\)-particle sector the
single-particle momenta run over
\[
    \mathcal K_{L,N}
    =
    \left\{
    \frac{2\pi}{L}(m+\theta_N):
    m=0,\ldots,L-1
    \right\},
    \qquad \theta_N\in\{0,\tfrac12\},
\]
where $\theta_N$ keeps track of the fermionic boundary condition
induced by the Jordan--Wigner transformation.  For odd $N$, one has
$\theta_N=0$, corresponding to periodic boundary conditions,
$
c_{L+1}=c_1,
$
whereas for even $N$, one has $\theta_N=\frac{1}{2}$, corresponding to
antiperiodic boundary conditions,
$
c_{L+1}=-c_1.
$.
The single-particle dispersion is
\begin{equation}
    \varepsilon(k)=-2w\cos k.
    \label{eq:dispersion}
\end{equation}

The \(N\)-particle ground state \(|\mathrm{GS}_N\rangle\) is obtained by filling
the \(N\) momenta closest to \(0\).  In the position basis
\(|x_1,\ldots,x_N\rangle=c_{x_1}^\dagger\cdots c_{x_N}^\dagger|0\rangle\) its wave
function is the Vandermonde determinant \(\det(e^{ik_a x_b})_{a,b=1}^N\), which up
to an overall phase is precisely the positive function \eqref{eq:Phi_def}:
\begin{equation}
    \langle x_1,\ldots,x_N|\mathrm{GS}_N\rangle
    =
    \frac{1}{\sqrt{Z_{L,N}}}\,\Phi_N(x_1,\ldots,x_N).
    \label{eq:gs_wavefunction_phi}
\end{equation}
Summing the \(N\) lowest energies yields the ground-state energy
\begin{equation}
    E_{L,N}
    =
    -2w\,
    \frac{\sin(\pi N/L)}{\sin(\pi/L)}.
    \label{eq:ground_state_energy}
\end{equation}

The identity \eqref{eq:gs_wavefunction_phi} is what links the two models.  Since
\(\Phi_N>0\) and \(H_{L,N}^{XX}\Phi_N=E_{L,N}\Phi_N\), conjugating
\(H_{L,N}^{XX}\) by the diagonal multiplication operator
\((\widehat\Phi_N f)(\eta)=\Phi_N(\eta)f(\eta)\) produces a Markov generator,
namely the ground-state (Doob) transform
\begin{equation}
    \mathcal L_{L,N}^{\mathrm{SDEP}}
    =
    -\widehat\Phi_N^{-1}
    \left(
    H_{L,N}^{XX}-E_{L,N}
    \right)
    \widehat\Phi_N ,
    \label{eq:doob_transform}
\end{equation}
equivalently
\begin{equation}
    H_{L,N}^{XX}
    =
    E_{L,N}
    -
    \widehat\Phi_N
    \mathcal L_{L,N}^{\mathrm{SDEP}}
    \widehat\Phi_N^{-1}.
    \label{eq:inverse_doob_transform}
\end{equation}
Indeed, the off-diagonal element of \(H_{L}^{XX}\) between two configurations
related by an allowed jump is \(-w\), so
\(-\Phi_N(\eta)^{-1}H^{XX}_{\eta,\eta'}\Phi_N(\eta')=w\,\Phi_N(\eta')/\Phi_N(\eta)\)
recovers the rate \eqref{eq:sdep_rates_abstract}, while the eigenvalue equation
fixes the diagonal part so that the escape rates are those of a genuine Markov
generator.  Through \eqref{eq:doob_transform}, the SDEP semigroup
\(e^{t\mathcal L_{L,N}^{\mathrm{SDEP}}}\) is conjugate to the free-fermion
evolution \(e^{-t(H_{L,N}^{XX}-E_{L,N})}\), whose single-particle dynamics is
carried by the propagator
\begin{equation}
    p_t^{(L,N)}(x)
    =
    \frac{1}{L}
    \sum_{k\in\mathcal K_{L,N}}
    e^{-ikx}\,e^{t\varepsilon(k)},
    \label{eq:finite_volume_kernel}
\end{equation}
the superscript \(N\) recording only the boundary-condition sector.

\subsection{Invariant measure}
\label{subsec:invariant_measure}

Because \(\Phi_N\) is the ground-state wave function
\eqref{eq:gs_wavefunction_phi}, its squared amplitude defines a natural
probability measure on \(\Omega_{L,N}\),
\begin{equation}
    \pi_{L,N}(\eta)
    =
    \frac{1}{Z_{L,N}}\,\Phi_N(\eta)^2,
    \label{eq:dyson_measure}
\end{equation}
with normalization \(Z_{L,N}\).  This is the discrete circular Dyson measure: it
describes \(N\) particles on the ring with logarithmic repulsion
\[
    -2\sum_{1\leq i<j\leq N}
    \log\left|
    \sin\left(\frac{\pi(x_j-x_i)}{L}\right)
    \right| ,
\]
and coincides with the distribution of eigenvalue angles of random unitary
matrices undergoing Dyson's Brownian motion.

The measure \eqref{eq:dyson_measure} is reversible for the SDEP.  Indeed, for any
allowed jump \(\eta\to\eta^{x,y}\),
\[
    \pi_{L,N}(\eta)\,c(\eta,\eta^{x,y})
    =
    \frac{w}{Z_{L,N}}\,\Phi_N(\eta)\,\Phi_N(\eta^{x,y})
\]
is symmetric under \(\eta\leftrightarrow\eta^{x,y}\), so detailed balance
\begin{equation}
    \pi_{L,N}(\eta)\,c(\eta,\eta')
    =
    \pi_{L,N}(\eta')\,c(\eta',\eta)
    \label{eq:detailed_balance}
\end{equation}
holds and \(\pi_{L,N}\) is stationary.  Equivalently, reversibility reflects the
Hermiticity of \(H_{L,N}^{XX}\): the transform \eqref{eq:doob_transform} makes
\(\mathcal L_{L,N}^{\mathrm{SDEP}}\) self-adjoint in \(L^2(\pi_{L,N})\).

The ground-state transform is the mechanism behind all the exact formulas of the
following sections.  Through it, time-dependent SDEP observables map to
free-fermion matrix elements between the XX ground state
\(|\mathrm{GS}_N\rangle\) and the initial configuration, evolved with the
propagator \eqref{eq:finite_volume_kernel}.  We make this dictionary precise in
the next section.

\subsection{ Observable dictionary between SDEP and free-fermions}
\label{sec:observable_dictionary}

The ground-state transform does more than identify the SDEP generator with the
XX Hamiltonian.  It also gives a dictionary between observables of the classical
exclusion process and matrix elements of free fermions.  We record this
dictionary before turning to explicit formulas.

Let \(F:\Omega_{L,N}\to\mathbb R\) be a classical observable.  We associate with
\(F\) the diagonal operator \(\widehat F\) on the \(N\)-particle fermionic Fock
space by
\[
    \widehat F|\eta\rangle=F(\eta)|\eta\rangle .
\]
In particular,
\[
    \eta_x
    \quad\longleftrightarrow\quad
    c_x^\dagger c_x,
\]
and, for distinct sites \(x_1,\ldots,x_m\),
\[
    \eta_{x_1}\cdots \eta_{x_m}
    \quad\longleftrightarrow\quad
    c_{x_1}^\dagger c_{x_1}\cdots
    c_{x_m}^\dagger c_{x_m}.
\]
For a set of occupied sites
\[
    S=\{x_1<\cdots<x_N\}\subset\mathbb T_L,
\]
we denote by
\[
    |S\rangle
    :=
    c_{x_1}^\dagger\cdots c_{x_N}^\dagger|0\rangle
\]
the corresponding fermionic occupation-basis state.
\begin{proposition}
\label{prop:observable_dictionary}
Let the SDEP start from a deterministic configuration
\(S\in\Omega_{L,N}\).  Then, for every observable
\(F:\Omega_{L,N}\to\mathbb R\),
\begin{equation}
    \mathbb E_S[F(\eta(t))]
    =
    \frac{
    \langle \mathrm{GS}_N|
    e^{tH^{XX}_{L,N}}
    \widehat F
    e^{-tH^{XX}_{L,N}}
    |S\rangle
    }{
    \langle \mathrm{GS}_N|S\rangle
    }.
    \label{eq:observable_dictionary}
\end{equation}
\end{proposition}

\begin{proof}
The ground-state transform gives
\[
    e^{t\mathcal L^{\mathrm{SDEP}}_{L,N}}
    =
    \widehat\Phi_N^{-1}
    e^{-t(H^{XX}_{L,N}-E_{L,N})}
    \widehat\Phi_N .
\]
Since \(F\) is diagonal, \(\widehat F\) commutes with \(\widehat\Phi_N\).  Using
that
\[
    \langle \mathrm{GS}_N|\eta\rangle
    =
    Z_{L,N}^{-1/2}\Phi_N(\eta),
\]
and that the ground-state energy cancels between numerator and denominator, one
obtains \eqref{eq:observable_dictionary}.
\end{proof}

For the density, Proposition~\ref{prop:observable_dictionary} gives
\begin{equation}
    \rho_S(x,t)
    =
    \mathbb E_S[\eta_x(t)]
    =
    \frac{
    \langle \mathrm{GS}_N|
    e^{tH^{XX}_{L,N}}
    c_x^\dagger c_x
    e^{-tH^{XX}_{L,N}}
    |S\rangle
    }{
    \langle \mathrm{GS}_N|S\rangle
    }.
    \label{eq:density_dictionary}
\end{equation}
Thus a classical density is converted into a free-fermion two-point function,
but with a mixed matrix element: the left state is the XX ground state, whereas
the right state is the deterministic configuration \(S\).

More generally, define the time-dependent (mixed) kernel
\begin{equation}
    K_t^S(x,y)
    =
    \frac{
    \langle \mathrm{GS}_N|
    e^{tH^{XX}_{L,N}}
    c_x^\dagger c_y
    e^{-tH^{XX}_{L,N}}
    |S\rangle
    }{
    \langle \mathrm{GS}_N|S\rangle
    }.
    \label{eq:time_dependent_mixed_kernel}
\end{equation}
Then
\[
    \rho_S(x,t)=K_t^S(x,x).
\]

Since \(H^{XX}_{L,N}\) is quadratic and both
\(\langle\mathrm{GS}_N|\) and \(|S\rangle\) are Slater states, Wick
factorization holds.  Therefore, for distinct sites
\(x_1,\ldots,x_m\),
\begin{equation}
    \mathbb E_S\left[
    \prod_{a=1}^m \eta_{x_a}(t)
    \right]
    =
    \det
    \left[
    K_t^S(x_a,x_b)
    \right]_{a,b=1}^m .
    \label{eq:determinantal_correlations}
\end{equation}
Thus, at each fixed time, the SDEP started from a deterministic configuration is
described by a determinantal kernel, although this kernel is in general
non-Hermitian because it comes from a mixed matrix element rather than from an
ordinary quantum expectation.

Finally, using the free-fermion evolution,
\[
    e^{tH^{XX}_{L,N}}c_x^\dagger e^{-tH^{XX}_{L,N}}
    =
    \sum_{u\in\mathbb T_L}p_t^{(L,N)}(x-u)c_u^\dagger,
\]
and
\[
    e^{tH^{XX}_{L,N}}c_y e^{-tH^{XX}_{L,N}}
    =
    \sum_{v\in\mathbb T_L}p_{-t}^{(L,N)}(v-y)c_v,
\]
the kernel becomes
\begin{equation}
    K_t^S(x,y)
    =
    \sum_{u,v\in\mathbb T_L}
    p_t^{(L,N)}(x-u)
    p_{-t}^{(L,N)}(v-y)
    K^S_0(u,v),
    \label{eq:kernel_evolution}
\end{equation}
where
\begin{equation}
    K^S_0(u,v)
    =
    \frac{
    \langle \mathrm{GS}_N|c_u^\dagger c_v|S\rangle
    }{
    \langle \mathrm{GS}_N|S\rangle
    }.
    \label{eq:initial_mixed_kernel}
\end{equation}
The next section evaluates \(K^S_0(u,v)\) explicitly for deterministic
configurations.

\section{Correlations with deterministic configurations}
\label{sec:mixed_correlations}

It remains to evaluate the initial kernel
\[
    K^S_0(x,y)
    =
    \frac{\langle \mathrm{GS}_N|c_x^\dagger c_y|S\rangle}
    {\langle \mathrm{GS}_N|S\rangle}
\]
for deterministic configurations \(S\). This is the only model-specific input
needed in the kernel evolution formula \eqref{eq:kernel_evolution}. We first
derive the exact finite-volume expression on the torus and then take its
infinite-volume limit at fixed particle number.

\subsection{Finite-volume mixed kernel}

Let \(S\subset\mathbb T_L\) be a deterministic set of \(N\) occupied sites.
The initial mixed kernel is given by a trigonometric Lagrange interpolation
formula.

\begin{proposition}
\label{prop:mixed_correlation_deterministic}
For \(x,y\in\mathbb T_L\), one has
\begin{equation}
    K^S_0(x,y)
    =
    \begin{cases}
    \displaystyle
    \prod_{z\in S\setminus\{y\}}
    \frac{
        \sin\left(\frac{\pi}{L}(x-z)\right)
    }{
        \sin\left(\frac{\pi}{L}(y-z)\right)
    },
    & y\in S, \\[1.4em]
    0,
    & y\notin S .
    \end{cases}
    \label{eq:mixed_correlation_formula}
\end{equation}
\end{proposition}
Notice that the right-hand side is the trigonometric Lagrange interpolation
polynomial associated with the set \(S\).  In particular, for fixed \(y\in S\),
the function
\[
    x\longmapsto K^S_0(x,y)
\]
equals \(1\) at \(x=y\) and vanishes at all other points of \(S\).  Thus
\[
    K^S_0(x,y)=\delta_{x,y},
    \qquad x,y\in S.
\]
Away from the initially occupied sites, however, \(K^S_0(x,y)\) is nonlocal.  This
nonlocal interpolation structure is the source of the explicit formulas below.

\begin{proof}
Let the occupied sites of \(S\) be \(s_1,\ldots,s_N\).  The XX ground state is a
Slater determinant built from the \(N\) lowest momenta.  We write these momenta
as
\[
    k_a
    =
    \frac{2\pi}{L}
    \left(
    a-\frac{N+1}{2}
    \right),
    \qquad a=1,\ldots,N,
\]
up to the harmless choice of boundary-condition convention.  Define the
\(N\times N\) matrix
\[
    V_{\alpha a}=e^{-ik_a s_\alpha},
    \qquad
    \alpha,a=1,\ldots,N.
\]
The overlap \(\langle \mathrm{GS}_N|S\rangle\) is proportional to
\(\det V\).

For \(y\in S\), say \(y=s_\beta\), we first compute the vector
\[
    u_a(y)
    :=
    \frac{
    \langle \mathrm{GS}_N|c^\dagger(k_a)c_y|S\rangle
    }{
    \langle \mathrm{GS}_N|S\rangle
    }.
\]
Using
\[
    c_{s_\alpha}^\dagger
    =
    \frac{1}{\sqrt L}\sum_{a=1}^N e^{-ik_a s_\alpha}c^\dagger(k_a)
    +\text{modes outside the filled Fermi sea},
\]
and using the canonical anticommutation relations inside the matrix element, we
obtain
\[
    \sum_{a=1}^N V_{\alpha a}u_a(y)
    =
    \sqrt L\,\delta_{\alpha,\beta}.
\]
Thus \(u_a(y)\) is determined by the inverse of the Vandermonde-type matrix
\(V\).

It remains to write this inverse explicitly.  Introduce the trigonometric
Lagrange function
\[
    \mathcal L_y(x)
    =
    \prod_{z\in S\setminus\{y\}}
    \frac{
    \sin\left(\frac{\pi}{L}(x-z)\right)
    }{
    \sin\left(\frac{\pi}{L}(y-z)\right)
    }.
\]
It satisfies
\[
    \mathcal L_y(z)=\delta_{y,z},
    \qquad z\in S.
\]
Moreover, as a function of \(e^{2\pi i x/L}\), it is a Laurent polynomial with
Fourier modes contained in the Fermi sea of the \(N\)-particle ground state.
Therefore its Fourier coefficients are precisely the entries of the inverse of
\(V\).  Consequently,
\[
    \frac{
    \langle \mathrm{GS}_N|c_x^\dagger c_y|S\rangle
    }{
    \langle \mathrm{GS}_N|S\rangle
    }
    =
    \mathcal L_y(x),
    \qquad y\in S.
\]
If \(y\notin S\), then \(c_y|S\rangle=0\), and the matrix element vanishes.  This
proves \eqref{eq:mixed_correlation_formula}.
\end{proof}

\subsection{Infinite-volume limit}

We shall often use Proposition~\ref{prop:mixed_correlation_deterministic} in
the regime where \(L\to\infty\) while the number of particles \(N\) remains
fixed.  In this limit, the torus becomes the infinite lattice \(\mathbb Z\), and
\[
    \sin\left(\frac{\pi}{L}(x-z)\right)
    \sim
    \frac{\pi}{L}(x-z).
\]
Therefore, for a deterministic set \(S\subset\mathbb Z\) with \(|S|=N\), the
mixed correlation becomes
\begin{equation}
    K_0^S(x,y)
    =
    \begin{cases}
    \displaystyle
    \prod_{z\in S\setminus\{y\}}
    \frac{x-z}{y-z},
    & y\in S, \\[1.2em]
    0,
    & y\notin S .
    \end{cases}
    \label{eq:mixed_correlation_infinite}
\end{equation}
This is the ordinary Lagrange interpolation formula on the lattice.
In particular, for the block initial condition
\begin{equation}
    S_N=\{1,2,\ldots,N\},
    \label{eq:block_initial_condition}
\end{equation}
one obtains
\begin{equation}
    K_{0}^{S_N}(x,y)
    =
    \begin{cases}
    \displaystyle
    \prod_{\substack{j=1\\ j\neq y}}^N
    \frac{x-j}{y-j},
    & y\in\{1,\ldots,N\}, \\[1.2em]
    0,
    & y\notin\{1,\ldots,N\}.
    \end{cases}
    \label{eq:block_mixed_correlation}
\end{equation}

Formula \eqref{eq:block_mixed_correlation} will be used in
Section~\ref{sec:block_melting} to specialize the general density evolution
to a compact block and to derive a genuinely finite representation of the
resulting melting profile.

\section{Exact density evolution}
\label{sec:exact_density_evolution}
By Proposition~\ref{prop:observable_dictionary}, the density is the diagonal
part of the time-dependent kernel:
\[
    \rho_S(x,t)=K_t^S(x,x).
\]
Using the free-fermion evolution formula \eqref{eq:kernel_evolution}, we obtain
the following exact expression.

\subsection{The finite-volume formula}
Using the free-fermion kernel \(p_t^{(L,N)}\) defined in
\eqref{eq:finite_volume_kernel}, yields the following result for evolution of the density profile in finite volume.
\begin{theorem}
\label{thm:exact_density_formula}
Let \(S\subset\mathbb T_L\) be a deterministic \(N\)-particle configuration.
Then the SDEP density at time \(t\) is
\begin{equation}
    \rho_S(x,t)
    =
    \sum_{u,v\in\mathbb T_L}
    p_t^{(L,N)}(x-u)\,
    p_{-t}^{(L,N)}(v-x)\,
    K^S_0(u,v),
    \label{eq:exact_density_formula_finite}
\end{equation}
where
\[
    K^S_0(u,v)
    =
    \frac{
    \langle \mathrm{GS}_N|c_u^\dagger c_v|S\rangle
    }{
    \langle \mathrm{GS}_N|S\rangle
    }
\]
is the mixed correlation from
Proposition~\ref{prop:mixed_correlation_deterministic}.  Equivalently,
\begin{equation}
    \rho_S(x,t)
    =
    \sum_{v\in S}
    \sum_{u\in\mathbb T_L}
    p_t^{(L,N)}(x-u)\,
    p_{-t}^{(L,N)}(v-x)\,
    \prod_{z\in S\setminus\{v\}}
    \frac{
    \sin\left(\frac{\pi}{L}(u-z)\right)
    }{
    \sin\left(\frac{\pi}{L}(v-z)\right)
    }.
    \label{eq:exact_density_formula_finite_explicit}
\end{equation}
\end{theorem}

Several elementary checks are immediate.  At \(t=0\),
\[
    p_0^{(L,N)}(x)=\delta_{x,0},
\]
and hence
\[
    \rho_S(x,0)=K_0^S(x,x)=\mathbf 1_{\{x\in S\}}.
\]
Moreover, the total mass is conserved.  Indeed,
\[
    \sum_{x\in\mathbb T_L}
    p_t^{(L,N)}(x-u)p_{-t}^{(L,N)}(v-x)
    =
    \delta_{u,v},
\]
and therefore
\[
    \sum_{x\in\mathbb T_L}\rho_S(x,t)
    =
    \sum_{u\in\mathbb T_L}K_0^S(u,u)
    =
    N.
\]

\subsection{Infinite-volume formula}

We shall also need the infinite-volume version of
Theorem~\ref{thm:exact_density_formula}.  Let \(S\subset\mathbb Z\) be a finite
set with \(|S|=N\).  In the limit \(L\to\infty\), the kernel
\eqref{eq:finite_volume_kernel} becomes
\begin{equation}
    p_t(x)
    =
    \int_{-\pi}^{\pi}
    \frac{dk}{2\pi}
    e^{-ikx}e^{-2wt\cos k}
    =
    I_x(-2wt),
    \label{eq:infinite_kernel_bessel}
\end{equation}
where \(I_x\) is the modified Bessel function of the first kind.  Thus
\[
    p_{-t}(x)=I_x(2wt).
\]

The density on \(\mathbb Z\) is therefore
\begin{equation}
    \rho_S(x,t)
    =
    \sum_{u\in\mathbb Z}
    \sum_{v\in S}
    I_{x-u}(-2wt)\,
    I_{v-x}(2wt)\,
    \prod_{z\in S\setminus\{v\}}
    \frac{u-z}{v-z}.
    \label{eq:exact_density_formula_infinite}
\end{equation}
Since \(I_{v-x}=I_{x-v}\), this may also be written as
\begin{equation}
    \rho_S(x,t)
    =
    \sum_{u\in\mathbb Z}
    \sum_{v\in S}
    I_{x-u}(-2wt)\,
    I_{x-v}(2wt)\,
    \prod_{z\in S\setminus\{v\}}
    \frac{u-z}{v-z}.
    \label{eq:exact_density_formula_infinite_symmetric}
\end{equation}

Equations \eqref{eq:exact_density_formula_infinite} and
\eqref{eq:exact_density_formula_infinite_symmetric} are exact for every finite
time \(t\) and every finite deterministic set \(S\).  They are the
infinite-volume analogues of the finite-torus formula
\eqref{eq:exact_density_formula_finite_explicit}.
\section{Moment identities and block melting}
\label{sec:block_melting}

We now investigate the evolution of the spatial moments of the density profile
for the infinite-volume SDEP. 

Let
\[
    X_1(t)<\cdots<X_N(t)
\]
denote the particle positions at physical time \(t\) for the process started
from the deterministic configuration \(S\), and let \(\mathbb E_S\) denote the
corresponding expectation.  For every nonnegative integer \(n\), define the
\(n\)-th moment of the density profile by
\begin{equation}
    \mu_n^S(t)
    :=
    \sum_{x\in\mathbb Z}x^n\rho_S(x,t).
    \label{eq:density_moment_definition}
\end{equation}
Equivalently, 
\begin{equation}
    \mu_n^S(t)
    =
    \mathbb E_S\left[
        \sum_{i=1}^N X_i(t)^n
    \right].
    \label{eq:density_moment_probabilistic_representation}
\end{equation}
The zeroth moment satisfies
$$
\mu_0^S(t)=N
$$
and expresses conservation of the total number of particles.  The first
moment determines the mean position of the particle cloud, whereas the second
and higher moments describe its spreading and provide increasingly detailed
information about the shape of the density profile.

Although the exact density formula obtained in the previous section is
explicit, computing \eqref{eq:density_moment_definition} directly leads to
infinite sums involving modified Bessel functions.  One of the main purposes
of this section is to show that the complete moment hierarchy admits a much
simpler finite-dimensional algebraic representation.  More precisely, the
time dependence of the moments will be encoded by an operator substitution
$$
X\longmapsto X+sD
$$
acting on the space of polynomials of degree strictly smaller than $N$, where $X$ and $D$ will be defined below. This
representation separates the evolution dynamics from the dependence on the
initial configuration and implies, in particular, that every moment is a
polynomial in time of bounded degree.

We first establish this representation for an arbitrary deterministic initial
set $S$ and derive the resulting identities for the first two moments.  We
then study the highest-time coefficient of the even moments and show that it
is independent of the detailed geometry of the initial configuration.  Its
large-$N$ behavior is governed by the Catalan numbers. Finally, we specialize the general results to the compact block \(S_N\)
defined in \eqref{eq:block_initial_condition}, for which we obtain a finite
representation of the exact melting profile and explicit formulas for the
first few centered even moments.

It is convenient throughout this section to introduce the dimensionless time
\begin{equation}
s=2wt.
\label{eq:dimensionless_time}
\end{equation}
We write $\rho_S(x,s)$ and $\mu_n^S(s)$ for the density and its moments at the
physical time $t=s/(2w)$.

The infinite-volume jump rates are
\begin{equation}
c_i^+(x_1,\ldots,x_N)
=
w\prod_{j\neq i}
\frac{x_i+1-x_j}{x_i-x_j},
\qquad
c_i^-(x_1,\ldots,x_N)
=
w\prod_{j\neq i}
\frac{x_i-1-x_j}{x_i-x_j}.
\label{eq:infinite_sdep_rates}
\end{equation}
A jump onto an occupied site has rate zero, as is already encoded in
\eqref{eq:infinite_sdep_rates}.  Indeed, if $x_i+1=x_j$ or $x_i-1=x_j$ for
some $j\neq i$, then the corresponding product contains a vanishing factor.

\subsection{Polynomial moment algebra}
\label{sec:polynomial_moment_algebra}

Let
\[
    S=\{a_1,\ldots,a_N\}\subset\mathbb Z
\]
be an arbitrary deterministic initial configuration, and let
\[
    \mathcal V_N=\mathbb C_{<N}[x]
\]
be the space of polynomials of degree strictly smaller than \(N\).  For
\(a\in S\), define the Lagrange basis polynomial
\begin{equation}
    \ell_a(x)
    =
    \prod_{\substack{b\in S\\b\neq a}}
    \frac{x-b}{a-b},
    \label{eq:general_lagrange_basis}
\end{equation}
and the interpolation map
\begin{equation}
    (P_Sf)(x)
    =
    \sum_{a\in S}f(a)\ell_a(x).
    \label{eq:general_interpolation_projector}
\end{equation}
$P_S$ is a projector on $\mathcal V_N$, i.e. 
\(P_Sf\in\mathcal V_N\), it acts as the identity on
\(\mathcal V_N\).

Introduce the shift operator
\[
    (Tf)(x)=f(x+1),
\]
and the operators
\begin{equation}
    B=\frac{T+T^{-1}}{2},
    \qquad
    D=\frac{T-T^{-1}}{2},
    \qquad
    (Xf)(x)=xf(x).
    \label{eq:BDX_operators}
\end{equation}
They satisfy
\begin{equation}
    [B,X]=D,
    \qquad
    [D,X]=B,
    \qquad
    [B,D]=0,
    \qquad
    B^2-D^2=1.
\label{eq:BDX_algebra}
\end{equation}
These relations are standard in finite operator calculus.  In this
terminology, \(D\) is the central-difference delta operator and
\[
    B=[D,X]
\]
is its Pincherle derivative
\cite{RotaKahanerOdlyzko1973,DimakisMullerHoissenStriker1996}.
Equivalently, the operators \(X,B,D\) form a realization of the complexified
\(1+1\)-dimensional Poincaré algebra, with quadratic Casimir
\(B^2-D^2=1\).  The new ingredient below is their finite-dimensional
compression by the interpolation projector \(P_S\), which encodes the
deterministic initial configuration.
The commutation relations above imply the following exact conjugation identity,
which is the algebraic origin of the substitution $X\mapsto X+sD$.
\begin{lemma}
\label{lem:X_conjugation}
For every $s\in\mathbb R$,
\begin{equation}
e^{sB}Xe^{-sB}=X+sD.
\label{eq:X_substitution}
\end{equation}
\end{lemma}

\begin{proof}
Define
$$
F(s)=e^{sB}Xe^{-sB}.
$$
Differentiating with respect to $s$ gives
$$
F'(s)
=
e^{sB}[B,X]e^{-sB}
=
e^{sB}De^{-sB}.
$$
Since $[B,D]=0$, the operator $D$ commutes with $e^{sB}$, and hence
$$
F'(s)=D.
$$
Together with the initial condition
$$
F(0)=X,
$$
this yields
$$
F(s)=X+sD,
$$
which proves \eqref{eq:X_substitution}.
\end{proof}
Let $e^{sB}(x,y)$ denote the kernel of $e^{sB}$ with respect to the canonical
basis of delta functions ${\delta_y}_{y\in\mathbb Z}$, namely
$$
e^{sB}(x,y)
=
\bigl(e^{sB}\delta_y\bigr)(x).
$$
This kernel can be written in terms of the modified Bessel function as follows:

\begin{lemma}
\label{lem:B_exponential_kernel}
For every $s\in\mathbb R$, the integral kernel of $e^{-sB}$ is
$$
\left(e^{-sB}\right)(x,u)
=
I_{x-u}(-s),
$$
where $I_n$ denotes the modified Bessel function of the first kind.
Equivalently, for every function $f$ for which the sum is well defined,
\begin{equation}
\left(e^{-sB}f\right)(x)
=
\sum_{u\in\mathbb Z}
I_{x-u}(-s)f(u).
\label{eq:B_exponential_kernel}
\end{equation}
\end{lemma}

\begin{proof}
Under the Fourier transform
$$
\widehat f(k)
=
\sum_{x\in\mathbb Z}f(x)e^{-ikx},
\qquad
k\in[-\pi,\pi],
$$
the operator $B$ acts by multiplication by $\cos k$:
$$
\widehat{Bf}(k)
=
\cos(k)\widehat f(k).
$$
Therefore,
$$
\widehat{e^{-sB}f}(k)
=
e^{-s\cos k}\widehat f(k).
$$
Taking the inverse Fourier transform gives
$$
\left(e^{-sB}f\right)(x)
=
\sum_{u\in\mathbb Z}
\left[
\int_{-\pi}^{\pi}
\frac{\mathrm dk}{2\pi},
e^{ik(x-u)}e^{-s\cos k}
\right]
f(u).
$$
Using the integral representation
$$
I_n(z)
=
\int_{-\pi}^{\pi}
\frac{\mathrm dk}{2\pi},
e^{ikn}e^{z\cos k},
$$
the quantity in brackets is $I_{x-u}(-s)$.  This proves
\eqref{eq:B_exponential_kernel}.
\end{proof}

\begin{theorem}[polynomial representation of the moments]
\label{prop:moment_operator_representation}
For every deterministic finite set \(S\subset\mathbb Z\), with \(|S|=N\), and
every \(n\geq0\),
\begin{equation}
    \mu_n^S(s)
    =
    \sum_{a\in S}
    \left[(X+sD)^n\ell_a\right](a)
    =  \operatorname{Tr}_{\mathcal V_N}
    \left[
    P_S(X+sD)^n
    \right].
\label{eq:moment_operator_representation}
\end{equation}
Consequently, \(\mu_n^S(s)\) is a polynomial in \(s\) of degree at most
\(\lfloor n/2\rfloor\).
\end{theorem}

\begin{proof}
The exact infinite-volume density formula
\eqref{eq:exact_density_formula_infinite_symmetric} can be written as
\begin{equation}
    \rho_S(x,s)
    =
    \sum_{a\in S}I_{x-a}(s)
    \left(e^{-sB}\ell_a\right)(x),
    \label{eq:general_density_operator_form}
\end{equation}
because the kernel of \(e^{-sB}\) is \(I_{x-u}(-s)\).  Therefore
\[
\begin{aligned}
    \mu_n^S(s)
    &=
    \sum_{a\in S}\sum_{x\in\mathbb Z}
    I_{a-x}(s)x^n\left(e^{-sB}\ell_a\right)(x)
    \\
    &=
    \sum_{a\in S}
    \left[e^{sB}X^ne^{-sB}\ell_a\right](a).
\end{aligned}
\]
Using \eqref{eq:X_substitution},
\[
    e^{sB}X^ne^{-sB}=(X+sD)^n,
\]
which proves the first equality in
\eqref{eq:moment_operator_representation}.

To obtain the trace formula, use the Lagrange basis
\(\{\ell_a:a\in S\}\) of \(\mathcal V_N\).  For any polynomial \(g\), the
coefficient of \(\ell_a\) in \(P_Sg\) is \(g(a)\).  Hence the trace of
\(P_S(X+sD)^n\) on \(\mathcal V_N\) is precisely the finite sum in the first
part of \eqref{eq:moment_operator_representation}.

Finally, expand \((X+sD)^n\) as a non-commutative polynomial.  The coefficient
of \(s^k\) is a sum of words containing \(k\) letters \(D\) and \(n-k\)
letters \(X\).  On polynomials, \(X\) raises the degree by one, whereas \(D\)
lowers it by at least one.  If \(k>n/2\), every such word strictly lowers
polynomial degree.  Its image already belongs to \(\mathcal V_N\), so \(P_S\)
does not change it, and its matrix in the monomial basis is strictly lower
triangular.  Its trace therefore vanishes.  This proves the degree bound.
\end{proof}
Explicit examples on how to use this theorem to compute the moments will be given in section \ref{sec:explicit_moments_catalan}.

\subsection{Universal first and second moments}

We next derive the first two moment identities directly from the jump rates.
The following interpolation lemma replaces the shifted identities by a single
formula valid for an arbitrary displacement.

\begin{lemma}
\label{lem:shifted_interpolation_sums}
Let \(x_1<\cdots<x_N\), let
\[
    \ell_i(z)=\prod_{j\neq i}\frac{z-x_j}{x_i-x_j},
    \qquad
    e_1=\sum_{i=1}^Nx_i,
\]
and let \(a\neq0\).  Then
\begin{equation}
    \sum_{i=1}^N\ell_i(x_i+a)=N,
    \label{eq:shifted_lagrange_sum_zero}
\end{equation}
and
\begin{equation}
    \sum_{i=1}^Nx_i\ell_i(x_i+a)
    =
    e_1+\binom{N}{2}a.
    \label{eq:shifted_lagrange_sum_one}
\end{equation}
\end{lemma}

\begin{proof}
The case \(N=1\) is immediate.  Assume \(N\geq2\), and set
\[
    Q(z)=\prod_{i=1}^N(z-x_i)
    =z^N-e_1z^{N-1}+\cdots
\]
and
\[
    R_a(z)=\frac{Q(z+a)-Q(z)}{a}.
\]
The polynomial \(R_a\) has degree at most \(N-1\), and
\[
    R_a(x_i)=\frac{Q(x_i+a)}{a}
    =Q'(x_i)\ell_i(x_i+a).
\]
Interpolating \(R_a\) at the nodes \(x_i\) gives
\[
    R_a(z)
    =
    \sum_{i=1}^NQ'(x_i)\ell_i(x_i+a)\ell_i(z).
\]
Moreover,
\[
    \ell_i(z)
    =
    \frac{z^{N-1}+(x_i-e_1)z^{N-2}+\cdots}{Q'(x_i)}.
\]
Comparing the coefficients of \(z^{N-1}\) yields
\eqref{eq:shifted_lagrange_sum_zero}.  Comparing those of \(z^{N-2}\), and
using
\[
    [z^{N-2}]R_a(z)
    =
    \binom{N}{2}a-(N-1)e_1,
\]
yields \eqref{eq:shifted_lagrange_sum_one}.
\end{proof}

A first consequence is that the total right and left jump rates are
configuration independent:
\begin{equation}
    \sum_{i=1}^Nc_i^+(x)=wN,
    \qquad
    \sum_{i=1}^Nc_i^-(x)=wN.
    \label{eq:total_left_right_rates}
\end{equation}
In particular, the total jump rate is \(2wN\), so the finite-particle process
is non-explosive.

\begin{theorem}
\label{thm:moment_identities}
For every deterministic finite initial configuration
\(S\subset\mathbb Z\), with \(|S|=N\), the first two density moments satisfy
\begin{equation}
    \mu_1^S(t)=\mu_1^S(0),
    \label{eq:first_moment_identity}
\end{equation}
and
\begin{equation}
    \mu_2^S(t)=\mu_2^S(0)+2wN^2t.
    \label{eq:second_moment_identity}
\end{equation}
\end{theorem}

\begin{proof}
Let
\[
    F_1(x)=\sum_{i=1}^Nx_i,
    \qquad
    F_2(x)=\sum_{i=1}^Nx_i^2.
\]
Since \(c_i^{\pm}(x)=w\ell_i(x_i\pm1)\), Lemma
\ref{lem:shifted_interpolation_sums} gives
\[
\begin{aligned}
    (\mathcal LF_1)(x)
    &=
    w\sum_{i=1}^N
    \left[\ell_i(x_i+1)-\ell_i(x_i-1)\right]
    =0.
\end{aligned}
\]
This proves \eqref{eq:first_moment_identity}.

For the second moment,
\[
\begin{aligned}
    (\mathcal LF_2)(x)
    =w\sum_{i=1}^N
    \bigl[
        (2x_i+1)\ell_i(x_i+1)
        +(-2x_i+1)\ell_i(x_i-1)
    \bigr].
\end{aligned}
\]
Using \eqref{eq:shifted_lagrange_sum_zero} and
\eqref{eq:shifted_lagrange_sum_one} with \(a=1\) and \(a=-1\), the sum in
brackets is
\[
\begin{aligned}
    &2\left[
      \sum_i x_i\ell_i(x_i+1)
      -\sum_i x_i\ell_i(x_i-1)
    \right]
    +\sum_i\left[\ell_i(x_i+1)+\ell_i(x_i-1)\right]
    \\
    &\hspace{2cm}
    =2N(N-1)+2N=2N^2.
\end{aligned}
\]
Therefore \(\mathcal LF_2=2wN^2\).  Taking expectations and integrating in
time proves \eqref{eq:second_moment_identity}.
\end{proof}

The random total position
\[
    F_1(t)=\sum_{i=1}^NX_i(t)
\]
is therefore a martingale, not a pathwise constant.  Since each jump changes it
by \(\pm1\), while the total jump rate is \(2wN\), its predictable quadratic
variation is \(2wNt\).  For a deterministic initial configuration,
\begin{equation}
    \operatorname{Var}F_1(t)=2wNt,
    \qquad
    \operatorname{Var}\left(\frac{F_1(t)}{N}\right)=\frac{2wt}{N}.
    \label{eq:center_of_mass_variance}
\end{equation}
Thus the expected center of mass is conserved, whereas the random center of
mass diffuses.

\subsection{Exact density profile for a compact block}
\label{sec:finite_block_melting}

We now specialize the general infinite-volume density formula to the compact
block \(S_N\) defined in \eqref{eq:block_initial_condition}. For \(v\in S_N\),
the corresponding Lagrange polynomial is
\begin{equation}
    \ell_v(u)
    =
    \prod_{\substack{j=1\\j\neq v}}^N
    \frac{u-j}{v-j}
    =
    \frac{(-1)^{N-v}}{(v-1)!(N-v)!}
    \prod_{\substack{j=1\\j\neq v}}^N(u-j).
    \label{eq:block_lagrange_factor}
\end{equation}
Specializing
\eqref{eq:exact_density_formula_infinite_symmetric} to \(S=S_N\), and writing
\(s=2wt\), gives
\begin{equation}
    \rho_N(x,s)
    =
    \sum_{u\in\mathbb Z}I_{x-u}(-s)
    \sum_{v=1}^N I_{x-v}(s)\ell_v(u).
    \label{eq:block_melting_formula}
\end{equation}

\begin{remark}
Since \(I_{x-u}(-s)\) is the kernel of \(e^{-sB}\), equation
\eqref{eq:block_melting_formula} may equivalently be written as
\[
    \rho_N(x,s)
    =
    \sum_{v=1}^N
    I_{x-v}(s)\bigl(e^{-sB}\ell_v\bigr)(x).
\]
Moreover, because \(B-1\) lowers polynomial degree by at least two and
\(\deg\ell_v=N-1\), the exponential truncates:
\[
    \rho_N(x,s)
    =
    e^{-s}
    \sum_{v=1}^N I_{x-v}(s)
    \sum_{k=0}^{\lfloor(N-1)/2\rfloor}
    \frac{(-s)^k}{k!}
    \bigl[(B-1)^k\ell_v\bigr](x).
\]
Thus, the infinite convolution in \eqref{eq:block_melting_formula} admits a
finite algebraic representation.
\end{remark}

\begin{remark}
For the block \(S_N=\{1,\ldots,N\}\), Theorem~\ref{thm:moment_identities}
implies that the mean position of the normalized density profile is conserved
and equal to
\[
    \bar x_N=\frac{N+1}{2}.
\]
Its variance is therefore
\begin{equation}
    \sigma_N^2(t)
    :=
    \frac{1}{N}
    \sum_{x\in\mathbb Z}
    (x-\bar x_N)^2\rho_N(x,t)
    =
    \frac{N^2-1}{12}+2wNt.
    \label{eq:block_width_growth}
\end{equation}
Thus, the squared width grows linearly at rate
\[
    \frac{d}{dt}\sigma_N^2(t)=2wN.
\]
This describes the spreading of the normalized one-point density and should
not be confused with the fluctuations of the random center of mass, whose
diffusion coefficient is \(w/N\).
\end{remark}

\subsection{Explicit centered moments}
\label{sec:explicit_moments_catalan}
The block is invariant under reflection about
\[
    \bar x_N=\frac{N+1}{2}.
\]
We therefore introduce the centered moments
\begin{equation}
    \mu_n^{(N)}(s)
    =
    \sum_{x\in\mathbb Z}
    (x-\bar x_N)^n\rho_N(x,s).
    \label{eq:moment_definition_centered}
\end{equation}
All odd centered moments vanish.  This definition applies to both even and odd
\(N\); for even \(N\), the translated lattice consists of half-integers.
Theorem~\ref{prop:moment_operator_representation} applies without change after
translation, with \(X\) replaced by multiplication by the centered coordinate.

For the second moment,
\[
    (X+sD)^2=X^2+s(XD+DX)+s^2D^2.
\]
The operator \(D^2\) strictly lowers degree and has zero trace.  Moreover,
\[
    Dx^m=mx^{m-1}+\text{terms of degree at most }m-3,
\]
so the diagonal coefficients of \(XD\) and \(DX\) on \(x^m\) are \(m\) and
\(m+1\), respectively.  Hence
\[
    \operatorname{Tr}_{\mathcal V_N}(XD+DX)
    =
    \sum_{m=0}^{N-1}(2m+1)=N^2.
\]
Together with the initial centered moment, this gives
\begin{equation}
    \mu_2^{(N)}(s)
    =
    \frac{N(N^2-1)}{12}+N^2s.
    \label{eq:mu2_centered_N}
\end{equation}

The next two even moments are
\begin{equation}
\begin{aligned}
    \mu_4^{(N)}(s)
    ={}&
    \frac{N(3N^4-10N^2+7)}{240}
    +\frac{N^2(N^2+1)}{2}s
    +N(2N^2+1)s^2,
\end{aligned}
\label{eq:mu4_centered}
\end{equation}
and
\begin{equation}
\begin{aligned}
    \mu_6^{(N)}(s)
    ={}&
    \frac{N(N^2-1)(3N^4-18N^2+31)}{1344}
    \\
    &+\frac{N^2(N^2+3)(3N^2+1)}{16}s
    \\
    &+\frac{N(9N^4+40N^2+11)}{4}s^2
    +5N^2(N^2+2)s^3.
\end{aligned}
\label{eq:mu6_centered}
\end{equation}
These expressions follow by expanding the finite-dimensional trace in
\eqref{eq:moment_operator_representation}.  In particular,
\begin{equation}
    [s]\mu_2^{(N)}(s)=N^2,
    \qquad
    [s^2]\mu_4^{(N)}(s)=2N^3+N,
    \qquad
    [s^3]\mu_6^{(N)}(s)=5N^4+10N^2.
    \label{eq:top_low_moments}
\end{equation}
The coefficients \(1,2,5\) are the leading large-\(N\) coefficients of these
highest-time terms, rather than the exact coefficients themselves.

\subsection{Universal leading coefficients and Catalan numbers}
In the following proposition, we show that for the even moments, the coefficient of the maximal power of time
depends only on the number of particles and not on their initial positions. In addition, these coefficients are given by the Catalan numbers.

\begin{proposition}
\label{prop:catalan}
Let \(S\subset\mathbb Z\) be any deterministic set with \(|S|=N\).  For fixed
\(r\geq1\), the coefficient
\[
    \kappa_{r,N}:=[s^r]\mu_{2r}^S(s)
\]
is independent of \(S\), and
\begin{equation}
    \kappa_{r,N}
    =
    C_rN^{r+1}+O_r(N^{r-1}),
    \qquad
    C_r=\frac{1}{r+1}\binom{2r}{r}.
    \label{eq:catalan_leading_coefficient}
\end{equation}
\end{proposition}

\begin{proof}
By Theorem~\ref{prop:moment_operator_representation}, the coefficient of
\(s^r\) is a sum over all balanced words of length \(2r\), containing \(r\)
letters \(X\) and \(r\) letters \(D\).  Let \(\mathcal W_r\) denote this set.
A balanced word maps every monomial \(x^m\), with \(m<N\), to a polynomial of
degree at most \(m\).  Its image therefore belongs to \(\mathcal V_N\), and
\(P_S\) acts as the identity.  Consequently,
\begin{equation}
    \kappa_{r,N}
    =
    \sum_{W\in\mathcal W_r}
    \operatorname{Tr}_{\mathcal V_N}W,
    \label{eq:catalan_word_trace_sum}
\end{equation}
which proves that \(\kappa_{r,N}\) is independent of \(S\).

Write a word in application order as
\[
    W=A_{2r}\cdots A_1,
    \qquad
    A_p\in\{X,D\},
\]
with \(A_1\) acting first.  Associate with \(W\) a path with an up step for
\(X\) and a down step for \(D\).  Let
\(h_1(W),\ldots,h_r(W)\) be the heights immediately before the successive down
steps.  Since
\[
    Dx^m=mx^{m-1}+\text{terms of degree at most }m-3,
\]
only the leading part of every \(D\) can contribute to the coefficient of
\(x^m\) in \(Wx^m\).  Therefore the diagonal coefficient is exactly
\begin{equation}
    [x^m]Wx^m
    =
    \prod_{q=1}^r\bigl(m+h_q(W)\bigr).
    \label{eq:catalan_diagonal_word}
\end{equation}

Let \(W^{\leftarrow}\) be the reversed word.  A down step of height \(h\) in
\(W\) corresponds to a down step of height \(1-h\) in
\(W^{\leftarrow}\).  Since reversal is an involution of \(\mathcal W_r\),
\begin{equation}
    \sum_{W\in\mathcal W_r}\sum_{q=1}^r
    \left(h_q(W)-\frac12\right)=0.
    \label{eq:height_reversal_cancellation}
\end{equation}
Expanding the exact product \eqref{eq:catalan_diagonal_word} around
\(m+\tfrac12\), and then summing over all words, gives
\[
    \sum_{W\in\mathcal W_r}
    \prod_{q=1}^r\bigl(m+h_q(W)\bigr)
    =
    \binom{2r}{r}\left(m+\frac12\right)^r
    +Q_{r-2}(m),
\]
where \(Q_{r-2}\) is a polynomial of degree at most \(r-2\) (and is absent
when \(r=1\)).  Hence
\[
    \kappa_{r,N}
    =
    \binom{2r}{r}
    \sum_{m=0}^{N-1}\left(m+\frac12\right)^r
    +O_r(N^{r-1}).
\]
The midpoint power sum satisfies
\[
    \sum_{m=0}^{N-1}\left(m+\frac12\right)^r
    =
    \frac{N^{r+1}}{r+1}+O_r(N^{r-1}),
\]
with no term of order \(N^r\).  This proves
\eqref{eq:catalan_leading_coefficient}.
\end{proof}

The appearance of Catalan numbers is, in fact, natural. In
\cite{ZahraDubailSchutz2025}, it was shown that, after the appropriate
large-time rescaling, the density profile converges to a semicircle law.
The even moments of the semicircle distribution are precisely the Catalan
numbers, up to the normalization fixed by its support. This is also consistent
with the random-matrix interpretation of the model. Indeed, for a fixed number
of particles evolving on an infinite lattice, the density becomes dilute at
large times. In this regime, the lattice constraint becomes asymptotically
irrelevant, and the system is expected to recover the behavior of the
continuous Dyson gas.

\begin{figure}
    \centering
\includegraphics[width=0.49\linewidth]{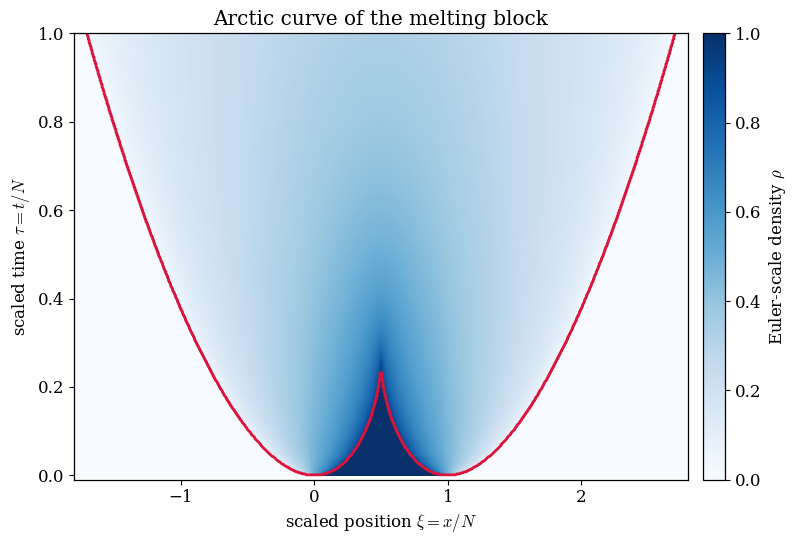}
\includegraphics[width=0.49\linewidth]{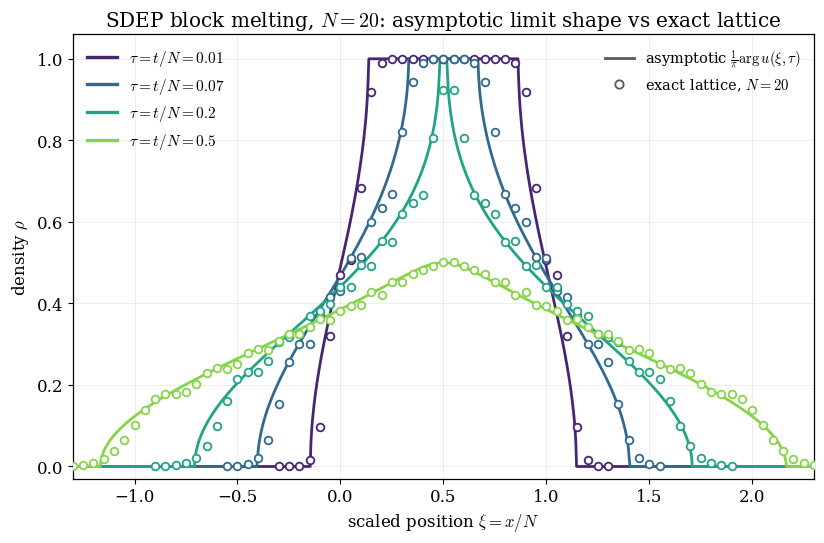}
\caption{%
Euler-scale melting of the block $S_N=\{1,\dots,N\}$, with $\xi=x/N$,
$\tau=t/N$ and $w=\tfrac12$.
\emph{Left:} phase diagram in the $(\xi,\tau)$ plane; colour is the limiting
density $\rho=\tfrac1\pi\arg u(\xi,\tau)$ (Theorem~\ref{thm:euler_block_profile}),
and the red arctic curve $\partial\mathcal L$ separates the liquid region
$0<\rho<1$ from the frozen phases $\rho=0$ (empty) and $\rho=1$ (full). The full
tongue closes at the cusp $(\tfrac12,\tfrac14)$;
\emph{Right:} density profiles at different times. Solid curves
are the limit shape $\tfrac1\pi\arg u(\xi,\tau)$; symbols are the exact
finite-lattice result~\eqref{eq:block_melting_formula} at $N=20$.
}
\label{fig:arctic_and_profiles}
\end{figure}

\section{Euler-scale limit and arctic curve of a single initial block}
\label{sec:block_euler_asymptotics}
The large-scale behavior of the SDEP was conjectured in
\cite{ZahraDubailSchutz2025} by solving the corresponding Euler-scale
conservation law. In particular, the melting of a single particle block was
studied, and the limiting density profile was expressed in terms of the roots
of a cubic equation. The nature of these roots distinguishes the liquid and
frozen regions. In this section, we recover these results directly from the
exact block-melting formula derived above, by means of a steepest-descent
analysis. The resulting phase diagram and representative density profiles are shown in
Figure~\ref{fig:arctic_and_profiles}.


We fix the microscopic timescale by setting
\[
    w=\frac12,
\]
so that the scaled time \(2wt\) coincides with the physical time \(t\). We
consider the Euler scaling
\begin{equation}
    x_N=\lfloor \xi N\rfloor,
    \qquad
    t_N=\tau N,
    \qquad
    \xi\in\mathbb R,
    \quad
    \tau>0.
    \label{eq:euler_scaling}
\end{equation}

For \((\xi,\tau)\in\mathbb R\times(0,\infty)\), consider the cubic equation
\begin{equation}
    \tau a^3
    +(2-2\xi-\tau)a^2
    +(2\xi-\tau)a
    +\tau
    =0.
    \label{eq:euler_cubic}
\end{equation}

\begin{definition}[Liquid region]
\label{def:euler_liquid_region}
Let \(\mathcal L\subset\mathbb R\times(0,\infty)\) be the set of points
\((\xi,\tau)\) for which \eqref{eq:euler_cubic} has exactly one real root and
one non-real conjugate pair. We denote by \(u(\xi,\tau)\) the root in the upper
half-plane and take
\[
    \arg u(\xi,\tau)\in(0,\pi).
\]
\end{definition}

We can now state the main result of this section.

\begin{theorem}[Euler-scale block profile]
\label{thm:euler_block_profile}
Let \(x_N=\lfloor \xi N\rfloor\) and \(t_N=\tau N\). If
\((\xi,\tau)\in\mathcal L\), then
\begin{equation}
    \lim_{N\to\infty}
    \rho_N(x_N,t_N)
    =
    \frac{\arg u(\xi,\tau)}{\pi},
    \label{eq:euler_density_limit}
\end{equation}
where \(u(\xi,\tau)\) is the upper-half-plane root of
\eqref{eq:euler_cubic}. The convergence is uniform for
\((\xi,\tau)\) in compact subsets of \(\mathcal L\).
\end{theorem}

The arctic curve is the boundary \(\partial\mathcal L\), along which two roots
of \eqref{eq:euler_cubic} coalesce. Equivalently, it is given by the vanishing
of the discriminant of the cubic:
\begin{equation}
    (2\xi-1)^4
    +
    \bigl(4\tau^2-20\tau-2\bigr)(2\xi-1)^2
    -64\tau^3+48\tau^2-12\tau+1
    =
    0.
    \label{eq:euler_arctic_curve}
\end{equation}

Equation~\eqref{eq:euler_arctic_curve} coincides with the arctic-curve
equation obtained from the hydrodynamic description in
\cite{ZahraDubailSchutz2025}, after shifting the spatial coordinate by
\(\xi\mapsto\xi-\frac12\) to account for the different position of the
initial block.  Indeed, the block considered here occupies
\(S_N=\{1,\ldots,N\}\), whereas the block in
\cite{ZahraDubailSchutz2025} is centered at the origin.  The large-\(\tau\)
limit provides a further connection with the moment analysis of the
preceding section.  In this regime, the arctic boundary satisfies
\[
    \left|\xi-\frac12\right|\sim 2\sqrt{\tau},
\]
and the Euler-scale density approaches the expanding semicircle profile
\[
    \rho(\xi,\tau)
    \sim
    \frac{1}{2\pi\tau}
    \sqrt{4\tau-\left(\xi-\frac12\right)^2}\,
    \mathbf 1_{\{|\xi-\frac12|<2\sqrt{\tau}\}}.
\]
Equivalently,
\[
    \sqrt{\tau}\,
    \rho\left(\frac12+\sqrt{\tau}\,u,\tau\right)
    \longrightarrow
    \frac{1}{2\pi}\sqrt{4-u^2}\,
    \mathbf 1_{\{|u|<2\}}.
\]
The even moments of this limiting distribution are the Catalan numbers:
\[
    \int_{-2}^{2}
    u^{2r}\frac{\sqrt{4-u^2}}{2\pi}\,du
    =C_r.
\]
This agrees with Proposition~\ref{prop:catalan}: setting \(t=N\tau\) and
\(w=\frac12\), its leading contribution gives
\[
    \frac{\mu_{2r}^{S_N}(N\tau)}{N^{2r+1}}
    \sim C_r\tau^r.
\]
Thus, the universal Catalan coefficients found algebraically in the
preceding section are the moment signature of the semicircle profile
emerging from the Euler-scale solution.

\subsection{Proof of Theorem~\ref{thm:euler_block_profile}}
\label{subsec:euler_block_profile_proof}

We begin by deriving an exact double-contour representation of the density.

\subsubsection{Exact double-contour representation}

The exact block formula \eqref{eq:block_melting_formula} reads
\begin{equation}
    \rho_N(x,t)
    =
    \sum_{x'\in\mathbb Z} I_{x-x'}(-t)
    \sum_{y=1}^{N} I_{x-y}(t)
    \prod_{\substack{j=1\\ j\neq y}}^{N}
    \frac{x'-j}{y-j}.
    \label{eq:euler_starting_block_formula}
\end{equation}
We use the convention
\begin{equation}
    I_n(t)
    =
    \int_{-\pi}^{\pi}\frac{dq}{2\pi}
    \exp\{t\cos q+inq\}.
    \label{eq:euler_bessel_convention}
\end{equation}

To account for the shift \(x_N-1\) appearing in the contour representation,
we set
\[
    \xi_N=\frac{x_N-1}{N}.
\]
Fix \(x\in\Z\) and \(t\ge0\), and set
\[
    f(y)=I_{x-y}(t),
    \qquad y\in S_N.
\]
Let \(P_{N-1}\) be the unique polynomial of degree at most \(N-1\) that
interpolates \(f\) on \(S_N\). Equivalently,
\[
    P_{N-1}(X)
    =
    \sum_{y=1}^N f(y)\ell_y(X),
\]
where \(\ell_y\) is the block Lagrange polynomial defined in
\eqref{eq:block_lagrange_factor}. Therefore,
\eqref{eq:euler_starting_block_formula} can be written as
\begin{equation}
    \rho_N(x,t)
    =
    \sum_{x'\in\Z} I_{x-x'}(-t) P_{N-1}(x').
    \label{eq:euler_density_interpolation_form}
\end{equation}
For consecutive nodes one has the Newton interpolation formula
\begin{equation}
    P_{N-1}(X)
    =
    \sum_{k=0}^{N-1}
    \binom{X-1}{k} (\Delta^k f)(1),
    \qquad
    \Delta f(y)=f(y+1)-f(y).
    \label{eq:euler_newton_interpolation}
\end{equation}

\begin{lemma}[Exact double-contour formula]
\label{lem:euler_exact_double_contour}
Let \(0<r<1\).  For every \(x\in\Z\) and \(t\ge0\),
\begin{align}
    \rho_N(x,t)
    & =
    \oint_{|\omega|=1}
    \frac{d\omega}{2\pi i\,\omega}
    \omega^{-(x-1)}
    \exp\left\{\frac{t}{2}\left(\omega+\omega^{-1}\right)\right\}
    \notag \\
    &\quad\times
    \frac{1}{2\pi i}
    \oint_{|z-1|=r}
    \frac{
        z^{x-1}
        \exp\left\{-\frac{t}{2}\left(z+z^{-1}\right)\right\}
    }{z-\omega}
    \left[
        1-
        \left(\frac{\omega-1}{z-1}\right)^N
    \right]dz .
    \label{eq:euler_exact_double_contour}
\end{align}
Both contours are positively oriented.  The apparent pole at \(z=\omega\) is
removable when the two contours meet.
\end{lemma}

\begin{proof}
From \eqref{eq:euler_newton_interpolation},
\[
    \rho_N(x,t)
    =
    \sum_{k=0}^{N-1}(\Delta^k f)(1)
    \sum_{x'\in\Z}I_{x-x'}(-t)\binom{x'-1}{k}.
\]
The coefficient representation
\[
    \binom{m}{k}
    =
    \frac{1}{2\pi i}
    \oint_{|z-1|=r}
    \frac{z^m}{(z-1)^{k+1}}\,dz
\]
gives, using the two-sided generating function for the modified Bessel
functions,
\begin{equation}
    \sum_{x'\in\Z}I_{x-x'}(-t)\binom{x'-1}{k}
    =
    \frac{1}{2\pi i}
    \oint_{|z-1|=r}
    \frac{
        z^{x-1}\exp\left\{-\frac{t}{2}(z+z^{-1})\right\}
    }{(z-1)^{k+1}}\,dz .
    \label{eq:euler_binomial_bessel_sum}
\end{equation}
On the other hand, the Fourier representation
\eqref{eq:euler_bessel_convention} gives
\[
    (\Delta^k f)(1)
    =
    \int_{-\pi}^{\pi}\frac{dp}{2\pi}
    \exp\{t\cos p+i(x-1)p\}
    (e^{-ip}-1)^k .
\]
Putting \(\omega=e^{-ip}\), this becomes an integral over the positively
oriented unit circle:
\[
    (\Delta^k f)(1)
    =
    \oint_{|\omega|=1}
    \frac{d\omega}{2\pi i\,\omega}
    \omega^{-(x-1)}
    \exp\left\{\frac{t}{2}(\omega+\omega^{-1})\right\}
    (\omega-1)^k .
\]
Inserting this and \eqref{eq:euler_binomial_bessel_sum}, the finite sum over
\(k\) is geometric:
\[
    \sum_{k=0}^{N-1}
    \frac{(\omega-1)^k}{(z-1)^{k+1}}
    =
    \frac{1}{z-\omega}
    \left[
        1-
        \left(\frac{\omega-1}{z-1}\right)^N
    \right].
\]
This proves \eqref{eq:euler_exact_double_contour}.
\end{proof}

The first term in the square brackets in
\eqref{eq:euler_exact_double_contour} will cancel against a residue contribution
from the second term.  This leaves a representation involving a single phase.
Write \(\rho_N=I_N^{(A)}-I_N^{(B)}\), where \(I_N^{(A)}\) is the contribution of
\(1\) and \(I_N^{(B)}\) is the contribution of
\(((\omega-1)/(z-1))^N\).  For fixed \(\omega\), the pole at \(z=\omega\) in
\(I_N^{(A)}\) contributes precisely when \(\omega\) lies inside the circle
\(|z-1|=r\).  The same pole appears in \(I_N^{(B)}\), with the same residue.
Consequently these two contributions cancel, and one obtains
\begin{equation}
    \rho_N(x,t)=-K_N(x,t),
    \label{eq:euler_rho_minus_K}
\end{equation}
where
\begin{align}
    K_N(x,t)
    & =
    \oint_{\Gamma_\omega}\frac{d\omega}{2\pi i}
    \frac{1}{\omega}
    \omega^{-(x-1)}
    \exp\left\{\frac{t}{2}\left(\omega+\omega^{-1}\right)\right\}
    (\omega-1)^N
    \notag \\
    &\quad\times
    \frac{1}{2\pi i}
    \oint_{\Gamma_z}
    \frac{
        z^{x-1}
        \exp\left\{-\frac{t}{2}\left(z+z^{-1}\right)\right\}
    }{(z-\omega)(z-1)^N}
    \,dz .
    \label{eq:euler_K_before_phase}
\end{align}
Here \(\Gamma_\omega\) is a positively oriented contour around \(0\),
\(\Gamma_z\) is a positively oriented contour around \(1\), and the two contours
are disjoint.  The deformation of the \(\omega\)-contour away from the point
\(1\) is harmless because the factor
\[
    (\omega-1)^N
    \operatorname*{Res}_{z=1}
    \left[
        \frac{
            z^{x-1}
            \exp\{-\frac{t}{2}(z+z^{-1})\}
        }{(z-\omega)(z-1)^N}
    \right]
\]
is holomorphic in \(\omega\) near \(1\).  Indeed, writing
\(z=1+\zeta\) and \(\omega=1+u\), the coefficient of \(\zeta^{N-1}\) in
\(u^N/(\zeta-u)\) is a polynomial in \(u\).

Under the scaling \eqref{eq:euler_scaling}, define the finite-\(N\) phase
\begin{equation}
    \Psi_N(a)
    =
    \frac{\tau}{2}\left(a+a^{-1}\right)
    -\xi_N\log a
    +\log(a-1),
    \label{eq:euler_phase_N}
\end{equation}
and the limiting phase
\begin{equation}
    \Psi_{\xi,\tau}(a)
    =
    \frac{\tau}{2}\left(a+a^{-1}\right)
    -\xi\log a
    +\log(a-1).
    \label{eq:euler_phase}
\end{equation}
The logarithms are chosen continuously along each contour segment.  Since the
integrands contain only integer powers, the expressions are single-valued; the
real part of \(\Psi_{\xi,\tau}\) is independent of these choices.  With this
notation, \eqref{eq:euler_K_before_phase} becomes
\begin{equation}
    K_N(x_N,t_N)
    =
    \oint_{\Gamma_\omega}\frac{d\omega}{2\pi i}
    e^{N\Psi_N(\omega)}
    \frac{1}{2\pi i}
    \oint_{\Gamma_z}
    \frac{e^{-N\Psi_N(z)}}{\omega(z-\omega)}\,dz .
    \label{eq:euler_K_phase_form}
\end{equation}

\subsubsection{Critical points and steepest-descent contours}

The critical points of the limiting phase satisfy
\[
    \Psi'_{\xi,\tau}(a)
    =
    \frac{\tau}{2}\left(1-a^{-2}\right)
    -\frac{\xi}{a}
    +\frac{1}{a-1}
    =0.
\]
After clearing denominators, this is precisely the cubic equation
\eqref{eq:euler_cubic} introduced above. Thus, for
\((\xi,\tau)\in\mathcal L\), the phase has one real critical point and the
non-real conjugate pair
\[
    u(\xi,\tau),
    \qquad
    \overline{u(\xi,\tau)}.
\]

The contour deformation underlying the asymptotic analysis is illustrated in Figure~\ref{fig:steepest_descent_geometry}.
For \((\xi,\tau)\in\mathcal L\), the two non-real critical points are
non-degenerate.  On compact subsets of \(\mathcal L\), the saddles remain a
positive distance away from each other and from the singularities \(0\) and
\(1\).  The following standard steepest-descent statement is the analytic input
needed below.

\begin{lemma}
\label{lem:euler_steepest_descent_contours}
Let \(D\Subset\mathcal L\).  Uniformly for \((\xi,\tau)\in D\), the contours
\(\Gamma_\omega\) and \(\Gamma_z\) in \eqref{eq:euler_K_phase_form} can be
deformed to contours \(\widetilde\Gamma_\omega\) and
\(\widetilde\Gamma_z\) with the following properties.
There is an oriented curve \(\Sigma\), from \(\overline{u(\xi,\tau)}\) to
\(u(\xi,\tau)\), along which the two deformed contours cross once.  After the
crossing contribution along \(\Sigma\) is removed, the remaining double integral
\(\widetilde K_N\) satisfies
\begin{equation}
    \widetilde K_N=O(N^{-1/2})
    \label{eq:euler_K_tilde_bound}
\end{equation}
uniformly for \((\xi,\tau)\in D\).
\end{lemma}

\begin{proof}
This is the usual two-contour steepest-descent construction for a pair of
conjugate non-degenerate saddles.  The \(z\)-contour is deformed through the
saddles along steepest-ascent directions for \(\Re\Psi_{\xi,\tau}\), while the
\(\omega\)-contour is deformed through the same saddles along steepest-descent
directions.  The two deformed contours cross along the level line \(\Sigma\)
connecting \(\bar u\) to \(u\).  Since \(D\Subset\mathcal L\), the saddles are
uniformly non-degenerate and stay away from \(0\) and \(1\), so the contours and
all estimates can be chosen uniformly in \((\xi,\tau)\in D\).

Away from small neighborhoods of the saddles there is a uniform
\(\delta>0\) such that
\[
    \Re\bigl(\Psi_{\xi,\tau}(\omega)-\Psi_{\xi,\tau}(z)\bigr)
    \le -\delta .
\]
Those parts of the double integral are exponentially small.  Near each saddle,
a quadratic expansion gives, in suitable local coordinates,
\[
    \Re\bigl(\Psi_{\xi,\tau}(\omega)-\Psi_{\xi,\tau}(z)\bigr)
    \le
    -c\bigl(|\omega-u|^2+|z-u|^2\bigr),
\]
and similarly near \(\bar u\), with \(c>0\) uniform on \(D\).  After the pole
crossing at \(z=\omega\) has been extracted, the remaining local integral is
bounded by the corresponding Gaussian estimate, giving \eqref{eq:euler_K_tilde_bound}.
The replacement of \(\Psi_N\) by \(\Psi_{\xi,\tau}\) changes these estimates only
by \(o(1)\), uniformly on \(D\), because \((\xi_N,\tau)\to(\xi,
\tau)\).
\end{proof}

\subsubsection{Completion of the proof}

\begin{proof}[Proof of Theorem~\ref{thm:euler_block_profile}]
Start from the exact representation \eqref{eq:euler_K_phase_form}.

Start from the exact representation \eqref{eq:euler_K_phase_form}.  Deform the
two contours as in Lemma~\ref{lem:euler_steepest_descent_contours}.  During this
deformation, the pole at \(z=\omega\) is crossed once along the curve
\(\Sigma\).  With \(\Sigma\) oriented from \(\bar u\) to \(u\), the residue of
\(1/[\omega(z-\omega)]\) gives
\begin{equation}
    \widetilde K_N
    =
    K_N
    +
    \int_\Sigma \frac{d\omega}{2\pi i\,\omega} .
    \label{eq:euler_crossing_relation}
\end{equation}
Along the crossing term the exponential factors cancel, since
\(e^{N\Psi_N(\omega)}e^{-N\Psi_N(\omega)}=1\).

The curve \(\Sigma\) is chosen with zero winding around the origin, so
\begin{equation}
    \int_\Sigma \frac{d\omega}{2\pi i\,\omega}
    =
    \frac{\arg u(\xi,\tau)}{\pi}.
    \label{eq:euler_arg_integral}
\end{equation}
Indeed, any continuous branch of the logarithm along \(\Sigma\) has endpoint
arguments \(-\arg u\) and \(\arg u\).  Lemma~\ref{lem:euler_steepest_descent_contours}
therefore implies \(\widetilde K_N\to0\), uniformly on compact subsets of
\(\mathcal L\).  Hence
\[
    K_N
    \longrightarrow
    -\frac{\arg u(\xi,\tau)}{\pi}.
\]
Together with \(\rho_N=-K_N\), as given by
\eqref{eq:euler_rho_minus_K}, this proves
\eqref{eq:euler_density_limit}.
\end{proof}

\begin{remark}[Frozen regions]
At the boundary of \(\mathcal L\), the conjugate pair of roots collides on the
real axis.  When the limiting root lies on the positive real axis, the angle
\(\arg u\) tends to \(0\), and the limiting density tends to \(0\).  When it lies
on the negative real axis, \(\arg u\) tends to \(\pi\), and the limiting density
tends to \(1\).  These are the two frozen phases adjacent to the liquid region.
\end{remark}

\begin{figure}[t]
    \centering
    \includegraphics[width=0.49\linewidth]{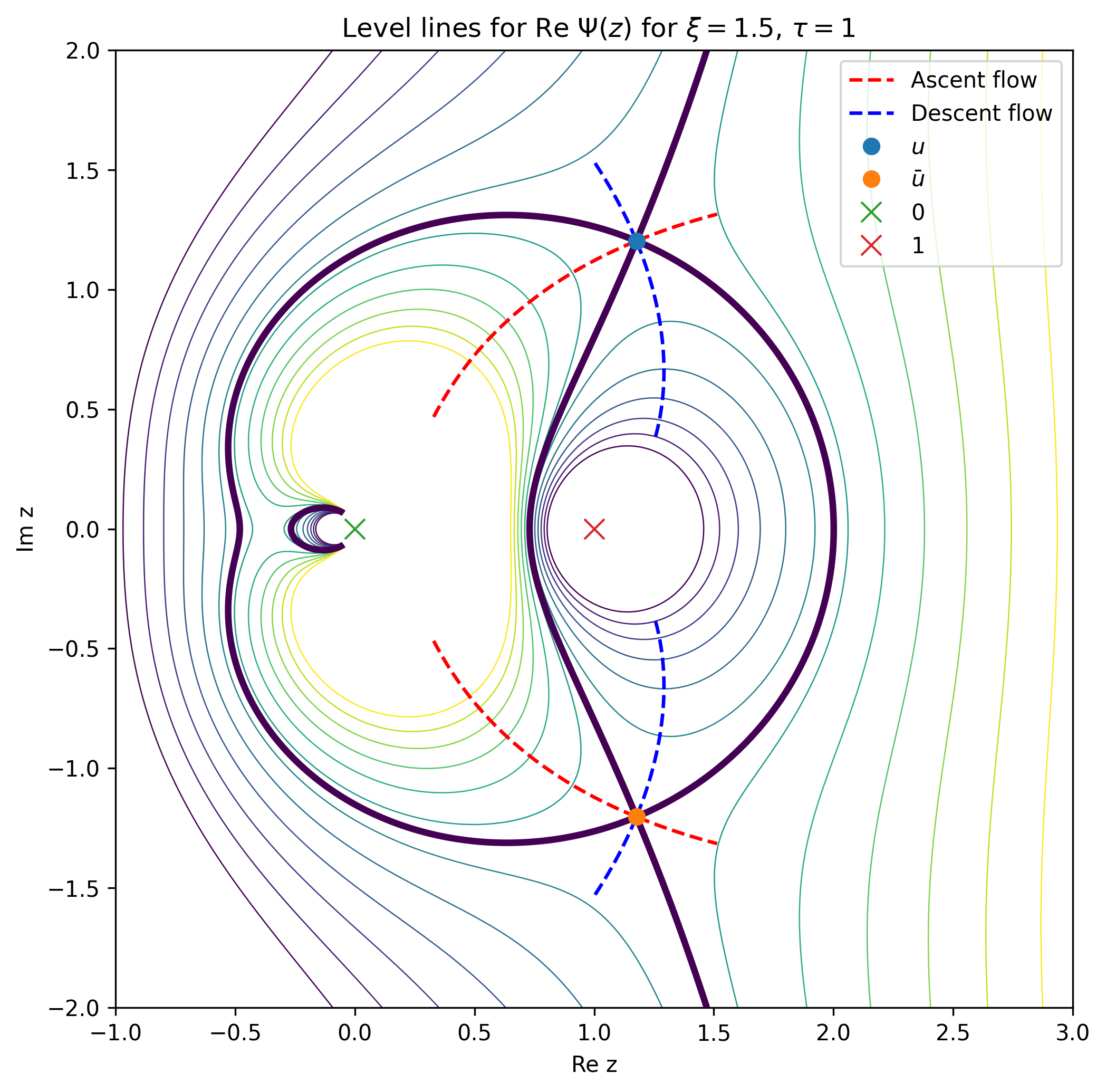}
    \includegraphics[width=0.49\linewidth]{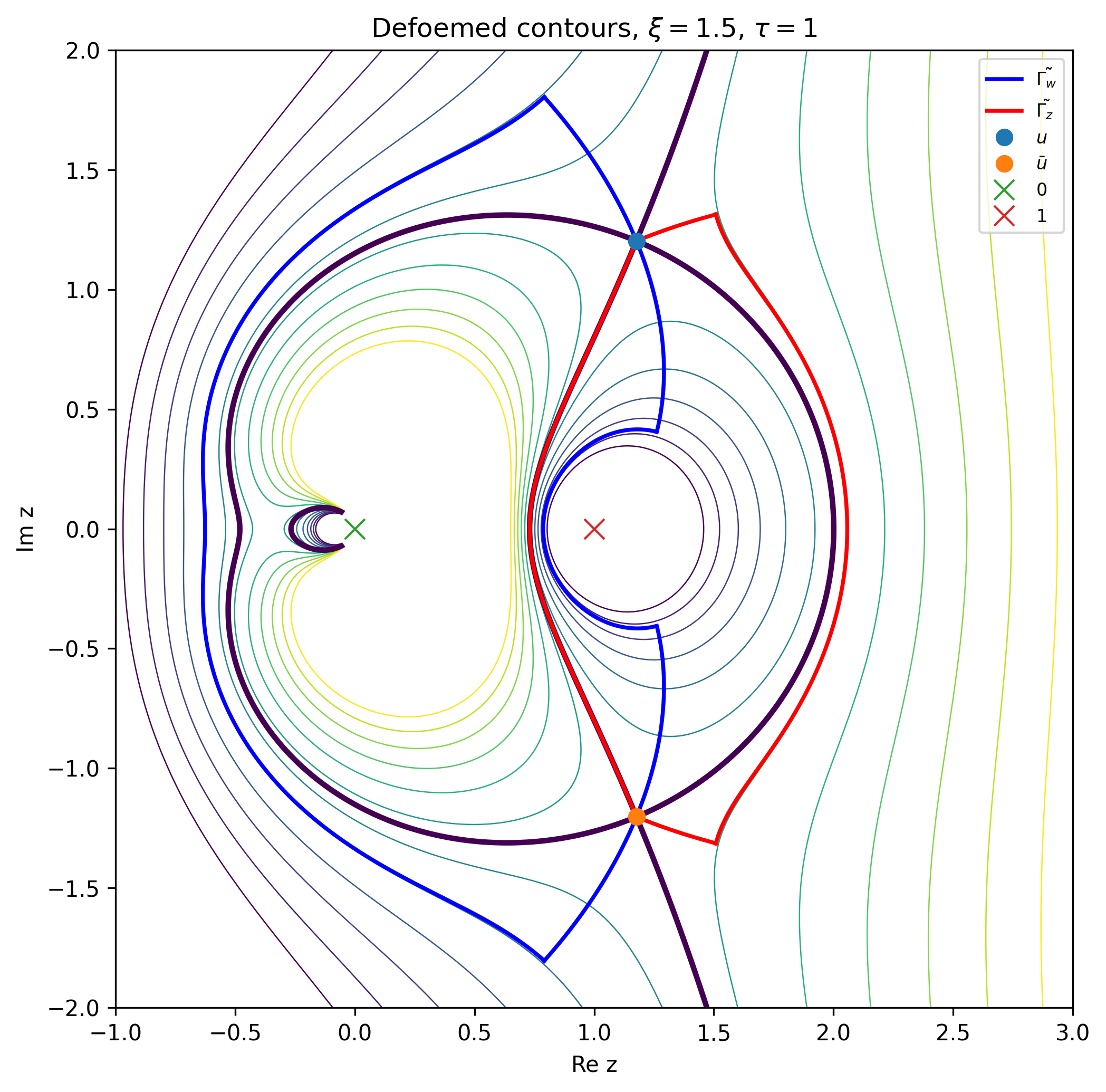}
    \caption{
    Steepest-descent geometry for the phase \(\Psi_{\xi,\tau}\), shown for a
    representative point \((\xi,\tau)\) in the liquid region.  Left: level
    curves of \(\Re\Psi_{\xi,\tau}\).  The thick purple curve is the equal-height
    locus through the relevant critical point, while dashed curves indicate
    ascent and descent directions.  Right: deformation of the original contours
    to the steepest-descent contours \(\widetilde\Gamma_\omega\) and
    \(\widetilde\Gamma_z\).  The Euler-scale density is produced by the crossing
    of these deformed contours.
    }
    \label{fig:steepest_descent_geometry}
\end{figure}

\section{Conclusion}
\label{sec:conclusion}

We have shown that the relaxation of the SDEP from deterministic initial
conditions can be described explicitly by combining the Doob-transform structure
of the process with the free-fermion representation of the XX chain.  This leads
to a mixed determinantal kernel and reduces the computation of the density to a
finite interpolation problem.  For arbitrary deterministic finite initial conditions, the resulting formulas
give a closed finite-dimensional algebra for all one-body density moments and a
universal Catalan pattern in their highest-time coefficients.  For block initial
conditions, they additionally give a genuinely finite representation of the
melting profile and an enhanced one-body spreading coefficient.

A main outcome of the analysis is the derivation of the Euler-scale profile for
the evolution of a single block.  This profile agrees with the one predicted by
the conjectured hydrodynamic description of Ref.~\cite{ZahraDubailSchutz2025}.  The present work
, therefore, provides an exact microscopic verification of that hydrodynamic
picture in a non-trivial deterministic setting.

It would be interesting to
extend the method developed here to arbitrary deterministic initial conditions,
where the interpolation kernel is still explicit, but the asymptotic analysis is
expected to be more involved. 
Another natural direction is the study of fluctuations around the deterministic hydrodynamic profile. In the liquid bulk, the dilute regime connects the SDEP to the continuous Dyson gas and to the non-local macroscopic fluctuation theory developed in \cite{DandekarKrapivskyMallick2024}. Its saddle-point action can be related, up to boundary terms, to an imaginary-time hydrodynamic action, providing a possible route to density and current large deviations. A natural next step is to extend this approach to the lattice regime, where finite-density effects can no longer be neglected.

\section*{Acknowledgments}
A.Z. thanks Dmitry Gangardt, Yasser Bezzaz, and Maksims Arzamasovs for
helpful discussions and for their hospitality at the University of
Birmingham. This work was funded by the ANR-PRME Uniopen project
(ANR-22-CE30-0004-01) and by FCT (Portugal) through project
UIDB/04459/2020 (doi:10.54499/UIDB/04459/2020) and grants
2020.03953.CEECIND and 2022.09232.PTDC.

\end{document}